\documentclass{llncs}
 
\usepackage{graphicx,amsmath, amsfonts}
\usepackage[english]{babel}
\usepackage[utf8]{inputenc}
\usepackage[T1]{fontenc}
 
\newcommand{\lcell}{\lceil}
\newcommand{\rcell}{\rceil}
\newcommand{\alphabetdf}{\mathfrak{A}_k}
\newcommand{\alphabet}{\mathfrak{A}_{2k}}
\newcommand{\obietnic}{2}
\newcommand{\eop}{\qed}
 
\author{Jerzy Marcinkowski, Jakub Michaliszyn}
\institute{Institute of Computer Science,\\ University Of Wroclaw,\\ ul. Joliot-Curie 15, 50-383 Wroclaw, Poland}
\title{The cost of being co-B\"uchi is nonlinear\thanks{Research supported by Polish Ministry of Science and Higher Education research project N206 022 
31/3660, 2006/2009.}}
 
\begin{document}
\maketitle
\pagestyle{headings}
 
\begin{abstract}
It is well known, and easy to see, that not each nondeterministic B\"uchi automaton on infinite words can be simulated by a  nondeterministic co-B\"uchi automaton. We show that in the cases when such a simulation is possible, the number of states needed for it can grow nonlinearly. More precisely, we show a sequence of -- as we believe, simple and elegant -- languages which witness the existence
of a  nondeterministic B\"uchi automaton with $n$ states, which can be simulated by a  nondeterministic co-B\"uchi automaton, but cannot be simulated by any  nondeterministic co-B\"uchi automaton with less than $c * n^{7/6}$ states for some constant $c$. This improves on the best previously known lower bound of $3(n-1)/2$.\footnote{Shortly before submitting this paper, we learned that a paper {\em Co-ing B\"uchi: Less Open, Much More Practical} by Udi Boker and Orna Kupferman was accepted for LICS 2009, and that it probably contains results that are similar to ours and slightly stronger. However, at the moment of our submission,
their paper was not published, and as far as we know, was not available on the Web.}
\end{abstract}
 
\section{Introduction}
\subsection{Previous work}
In 1962 B\"uchi was the first to introduce finite automata on infinite words. He needed them  to solve some fundamental decision problems in mathematics and logic (\cite{B62}, \cite{M66} \cite{R69}).  They  became a popular area of research due to their elegance and the tight relation between automata on infinite objects and monadic second-order logic. Nowadays, automata are seen as a very useful tool in verification and specification of nonterminating systems. This is why the complexity of problems concerning automata has recently been considered a hot topic (e. g. \cite{kv97}, \cite{AK08}).
 
To serve different applications, different types of automata were introduced. In his proof of the decidability of the satisfiability of S1S, B\"uchi introduced nondeterministic automata on infinite words (NBW), which are a natural tool to model things that happen infinitely often. In a B\"uchi automaton, some of the states are  {\it accepting} and a run on an infinite word is accepting if and only if  it visits some accepting state infinitely often (\cite{B62}). Dually, a run of co-B\"uchi automaton (NCW) is accepting if and only if it visits non-accepting states only finitely often. There are also automata with more complicated accepting conditions --  most well known of them are parity automata, Street automata, Rabin automata and Muller automata. 
 
As in the case of finite automata on finite words,  four basic types of transition relation can be considered: {\it deterministic, nondeterministic, universal} and {\it alternating}. {\bf In this paper, from now on, we only consider nondeterministic automata}. 
 
The problem of comparing the power of different types of automata is well studied and understood. For example it is easy to see that not every language that can be recognized by a  B\"uchi automaton on infinite words (such languages are called $\omega$-regular languages) can be also recognized by a  co-B\"uchi automaton. The most popular example of $\omega$-regular language that cannot be expressed by NCW is the language $L=\{w | w$ has infinitely many 0's$\}$ over the alphabet $\{0, 1\}$. On the other hand, it is not very hard to see that every language that can be recognized by a co-B\"uchi automaton is $\omega$-regular.
 
As we said, the problem of comparing the power of different types of automata is well studied. But we are quite far from knowing everything about the number of states needed to simulate an automaton of one type by an automaton of another type -- see for example the survey \cite{K07} to learn about the open problems in this area. 
 
In this paper we consider the  problem of the cost of simulating a B\"uchi automaton on infinite words by a co-B\"uchi automaton (if such NCW exists), left open in \cite{K07}:  given a number $n\in {\cal N}$,  for what $f(n)$ can we be sure that   every nondeterministic B\"uchi automaton with no more than $n$ states, which can be simulated by a co-B\"uchi automaton, can be  simulated by a co-B\"uchi automaton with at most $f(n)$ states?
 
There is a large gap between the known upper bound and the known lower bound for such a translation. The best currently known translation goes via intermediate deterministic Street automaton, involving exponential blowup of the number of states (\cite{S88}, \cite{AK08}). More precisely, for NBW with $n$ states, we get NCW with $2^O(n \log n)$ states. For a long time the best known lower bound for $f$  was nothing more than the  trivial bound $n$. In 2007 it was shown that there is an NBW with equivalent NCW such that there is no NCW equivalent to this NBW on the same structure (\cite{KMM04}). The first non-trivial (and the best currently known) lower bound is linear -- the result of  \cite{AK08} is  that,  for each $n\in {\cal N}$, there  exists a NBW with $n$ states such that there is a $NCW$ which recognizes the same language, but every such  $NCW$ has at least $3(n - 1)/2$ states.
 
There is a good  reason why it is  hard to show a lower bound for the above problem. The language (or rather, to be more precise, class of languages) used to show such a bound has to be hard enough to be expressed by a co-B\"uchi automaton, but on the other hand not too hard, because some (actually, most of) $\omega$-regular languages cannot be expressed by NCW at all. The idea given in the proof of $3(n - 1)/2$ lower bound in \cite{AK08}, was to define a language which can be easily split into parts that can by recognized by a NBW but cannot by recognized by a NCW. The language they used was $L_k = \{w \in \{0, 1\}^{\omega} | $ both $0$ and $1$ appear at least $k$ times in $w\}$. Let $L_k^i = \{w \in \{0, 1\}^{\omega} | i$ appears infinitely often in $w$ and $(1-i)$ appear at least $k$ times in $w\}$, then it is easy to see that $L_k = L_k^0 \cup L_k^1$, $L_k^i$ can be recognized by B\"uchi automata of size  $k$ and $L_k^i$ cannot by recognized by any co-B\"uchi automaton. It is still, however, possible to built a NCW that recognizes $L_k$ with $3k+1$ states and indeed, as it was proved in \cite{AK08},  every NCW  recognizing $L_k$ has at least $3k$ states.
 
\subsection{Our contribution -- a nonlinear lower bound}
In this paper we give a strong improvement of the lower bound from \cite{AK08}. We show that, for every integer $k$, there is a language $L_k$ such that $L_k$ can be recognized by an NBW with $\Theta(k^2)$ states, whereas every NCW that recognizes this language has at least $\Theta(k^{7/3})$ states. Actually, the smallest NCW we know, which recognizes $L_k$, has  $\Theta(k^{3})$ states, and we believe that this automaton is indeed minimal.  In the terms of function $f$ from the above subsection this means that $f$ equals at least $cn^{7/6}$ for some $c$ (this is since $n$ is $\Theta(k^2)$) and, if our conjecture concerning the size of a minimal automaton for $L_k$ is true,    $f$ would equal at least $cn^{3/2}$ for some constant $c$.
 
The technical part of of this paper is organized as follows. In subsection \ref{preliminaria} we give some basic definitions. In subsection \ref{Lk}  the definition of the language $L_k$ is presented. Also in this section  we show how this language can be  recognized by a  B\"uchi automaton with $O(k^2)$ states and how  $L_k$ can be recognized by a co-B\"uchi automaton with $O(k^3)$ states. The main theorem, saying that 
every co-B\"uchi automaton that recognizes $L_k$ has at least $\Theta(k^{7/3})$  states is formulated in the end of subsection \ref{Lk} and the rest of the paper is devoted to its proof.

\section{Technical Part}
\subsection{Preliminaries}\label{preliminaria}
A \emph{nondeterministic $\omega$-automaton} is a quintuple $\langle \Sigma, Q, q_0, \delta , \alpha\rangle$, where $\Sigma$ is an alphabet, $Q$ is a set of states, $q_0\in Q$ is an initial state, $\delta \subseteq Q \times \Sigma \times Q$ is a transition relation and $\alpha \subseteq Q$ is an accepting condition.
 
A \emph{run} of $\omega$-automaton over a word $w=w_1w_2\dots$ is a sequence of states $q_0q_1q_2\dots$ such that for every $i \geq 0$, $q_i\in Q$ and $\langle q_i, w_{i+1}, q_{i+1}\rangle \in \delta$. 
 
Depending on type of the automaton we  we have different definitions of \emph{accepting run}. For a B\"uchi automaton, a run is accepting, if it visits some state from accepting condition $\alpha$ infinitely often. In the case of a co-B\"uchi automaton, a run is accepting, if only states from the set $\alpha$  are visited infinitely often in this run. For a given nondeterministic $\omega$-automaton $\cal{A}$ and a given word $w$, we say that $\cal{A}$ \emph{accepts}   $w$ if there exists an accepting run of $\cal{A}$ on $w$. The words accepted by $\cal{A}$ form the language of $\cal{A}$, denoted by $L(\cal{A})$.
 
We say that a co-B\"uchi automaton ${\cal{A}}=\langle \Sigma , Q, q_0, \delta , \alpha\rangle$ is in the {\bf normal form} iff for each $\langle q, a, q'\rangle \in \delta$ if $q$ is in $\alpha$, then also $q'$ is in $\alpha$. Note that for a given NCW 
${\cal{A}}=\langle \Sigma, Q, q_0, \delta , \alpha\rangle$
the automaton ${\cal{A'}}=\langle \Sigma, Q', \langle q_0, 0\rangle, \delta', \alpha\times\{1\}\rangle$, where 
$Q'=Q\times\{0\} \cup \alpha \times \{1\}$ and 
$\delta'=\{\langle \langle q, i\rangle, a, \langle q', j \rangle\rangle \; | \; \langle q, a, q'\rangle \in \delta \wedge i \leq j 
\wedge  \langle q, i\rangle,\langle q', j \rangle\in Q' \}$ 
is in the normal form, recognizes the same language and has at most $2|Q|$ states.
 
For a given word $w=w_1, w_2, \dots$, let $w[i, j]=w_i, w_{i+1}, \dots, w_j$ and let $w[i, \infty]=w_i, w_{i+1}, \dots$. 
 
An accepting run $q_0, q_1, \dots$ on a word $w=w_1, w_2, \dots$ of a co-Buchi automaton in the normal form is called \emph{shortest}, if it reaches
an accepting state as early as possible, that is if for each accepting run $p_0, p_1, \ldots$ of this automaton on $w$ it holds that 
if $p_i\in \alpha$ then also $q_i \in \alpha$.
 
 
\subsection{Languages $L_k$ and their automata}\label{Lk}
Let $k \geq 64 $ be a -- fixed -- natural number and let $\alphabetdf=\{1,2\ldots k\}$. 
The set $\Sigma_k=\{a,\bar{a} \;| \: a\in \alphabet\}$ will
be the alphabet of our language $L_k$. 
 
Let us begin with the following informal interpretation of words 
over  $\Sigma_k$. Each symbol $j\in \Sigma_k$ should be read as ``agent $j$ makes a promise''.
Each symbol $\bar{j}\in \Sigma_k$ should be read as ``$j$ fulfills his  promise''. The language $L_k$ consists (roughly speaking)
of the words in which there is someone who at least $\obietnic k$ times fulfilled  his promises, but there are also
promises which were never fulfilled.
 
To be more formal: 
 
\begin{definition}
For a word $w\in \Sigma_k^\omega$, where $w=w_1w_2\ldots $ and $i\in \cal N$, define the interpretation $h_i(w)$ as:
 
\begin{itemize}
\item $h_i(w) = \sharp$ ~if $w_i \in \alphabet$ and $\bar{w_i}$ occurs in $w[i+1,\infty]$; (it is the fulfillment that counts, not a promise).
 
\item $h_i(w) = 0 $ if $w_i \in \alphabet$ and $\bar{w_i}$ does not occur in $w[i+1,\infty]$; (unfulfilled promises are read as 0).
 
\item Suppose $w_i = \bar{s}$ for some $s\in \alphabet$. Then  $h_i(w)=s$ if there is $j<i$ such that $w_j=s$ and $\bar{s}$ does not occur in the word
$w[j,i-1]$, and  $h_i(w)=\sharp$ if there is no such $j$ (one first needs to make a promise, in order to fulfill it).
 
\end{itemize}
 
The interpretation $h(w)$ is now defined as the infinite word $h_1(w)h_2(w)\ldots $.
\end{definition}
 
Now we are ready to formally define the language  $L_k$:
 
\begin{definition}
$L_k$ is the set of such words $w\in \Sigma_k^\omega$ that:
 
\begin{itemize}
\item
either there is at least one $0$ in $h(w)$ and there exists $s\in \alphabet$ which occurs at least $\obietnic k$ times in $h(w)$,
\item
or there exists $i$ such that $h_j(w)=\sharp$ for all $j>i$. 
\end{itemize}
\end{definition}

 \begin{figure}
  \begin{center}
  \includegraphics[width = \textwidth]{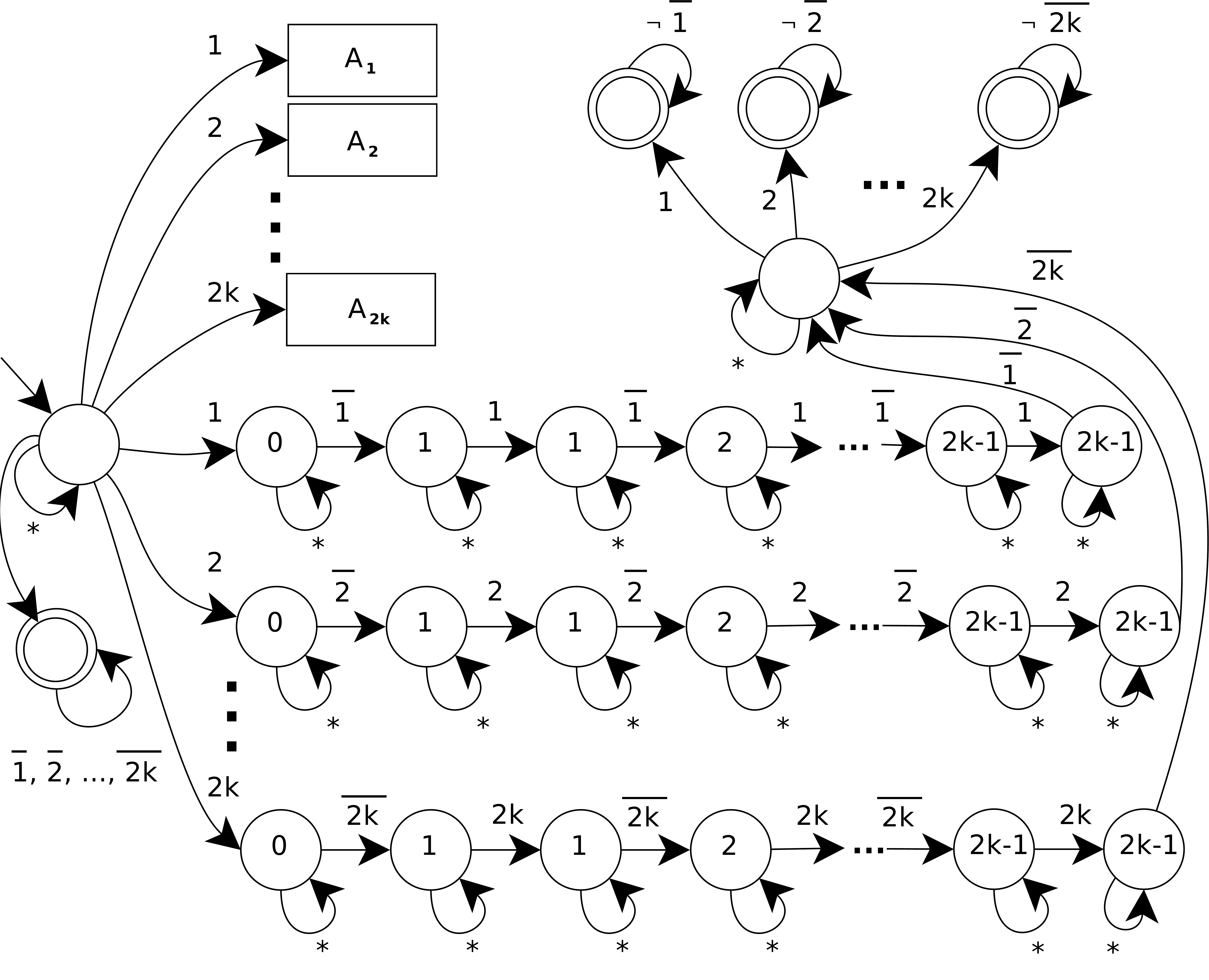}
  \end{center}
  \caption{The $\omega$-automaton recognizing $L_k$ -- all differences between NBW version and NCW version are in the body of $A_i$ (fig. \ref{f-ai}). Label $\neg i$ stands (for better readability) for alternative of every label except $i$.}
 \label{f-NBW}
 \end{figure}
 \begin{figure}
  \begin{center}
  \includegraphics[width = 270pt]{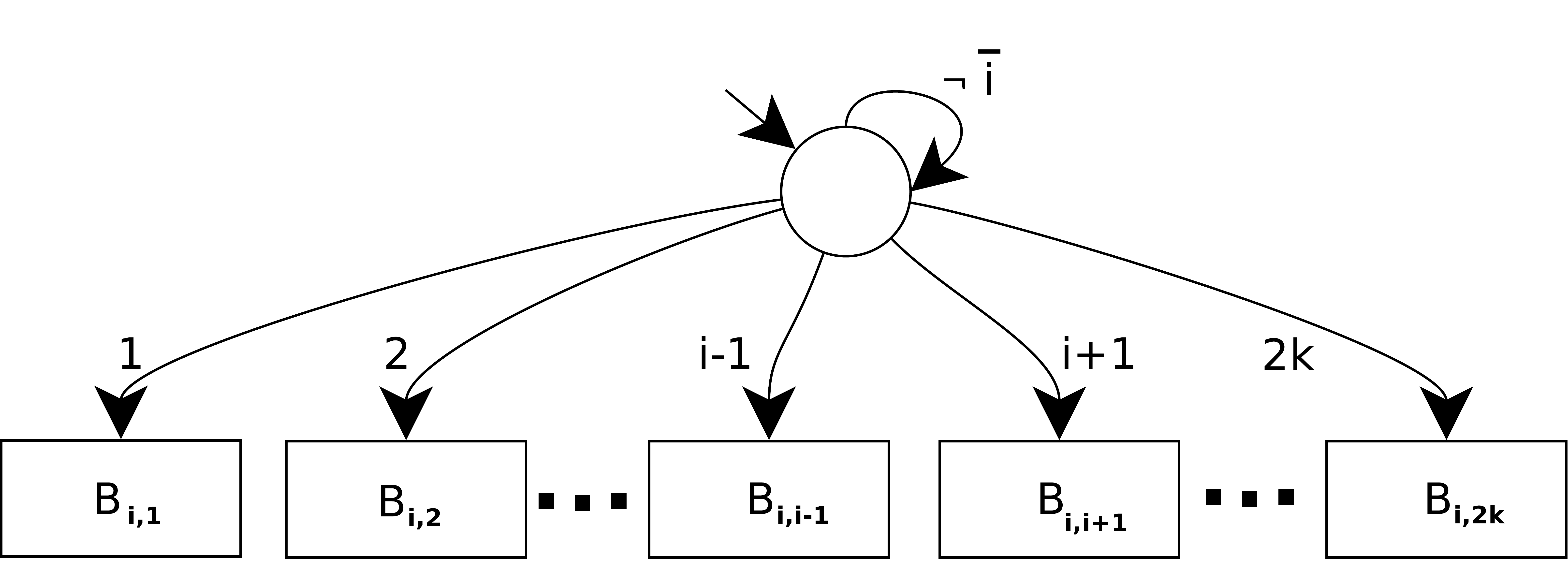}
  \end{center}
  \caption{Automaton $A_i$ -- the same for NBW and NCW, modulo the body of $B_{i, j}$ (fig. \ref{f-bij-buchi} for NBW and fig. \ref{f-bij-cobuchi} for NCW)}
 \label{f-ai}
 \end{figure}


It is easy to see that each $w\in \Sigma_k^\omega$  satisfies at least one of the following three conditions:
there is $s\in \alphabet$ such that $h_i(w)=s$ for infinitely many numbers $i$, or there are infinitely many occurrences of $0$ in $h(w)$, or
there is only a finite number of occurrences of symbols from $\alphabet$ in $w$. Using this observation, we can represent $L_k$ in the following way:
 
\begin{eqnarray}
\nonumber L_k & = & \{ viw \ | v \in \Sigma_k^* \wedge \\
& ( &\nonumber (i \in \alphabet \wedge w \in (\Sigma_k\setminus\{\overline{i}\})^{\omega} \wedge \exists j \in \alphabet\;\; h_m(viw)=j \text{ for infinitely many} \\
& & \text{ numbers } m) \label{inf0} \\
& \nonumber \vee & (i \in \alphabet \wedge w \in (\Sigma_k\setminus\{\overline{i}\})^{\omega} \wedge \exists j \in \alphabet \;\; h_m(viw)=j \text{ for at least }\\ 
& &\obietnic k \text{ numbers } m \text{ such that } m\leq |v|) \label{inf1} \\
& \vee & (w \in \{\overline{1}, \dots, \overline{2k}\}^{\omega}))\} \label{noinf}
\end{eqnarray}
 
 \begin{figure}
  \begin{center}
  \includegraphics[width = 160pt]{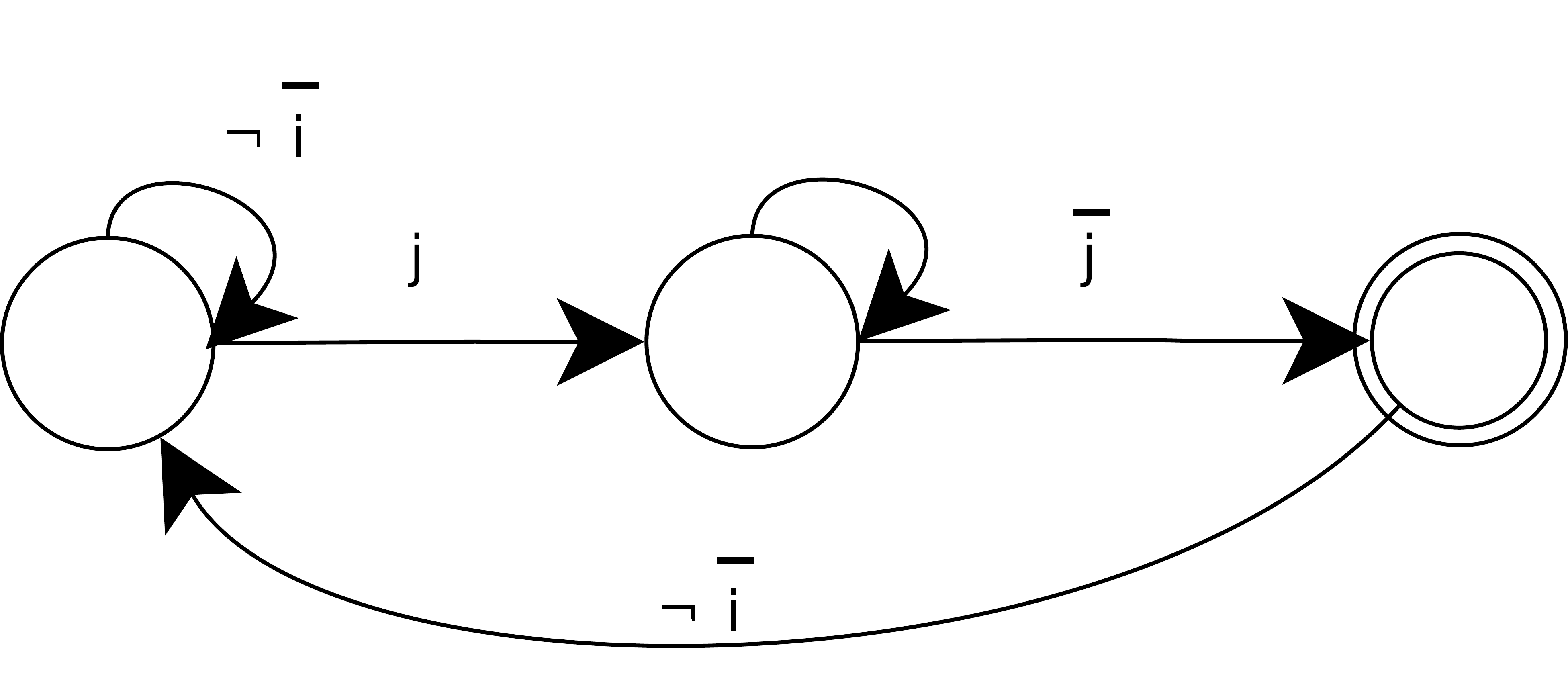}
  \end{center}
  \caption{Automaton $B_{i,j}$ in the NBW case}
 \label{f-bij-buchi}
 \end{figure}

Keeping in mind the above representation, it is easy to build a small NBW recognizing $L_k$ (see Figures \ref{f-NBW}, \ref{f-ai} and \ref{f-bij-buchi}). The accepting state on the bottom left of  Figure \ref{f-NBW} checks if  condition (\ref{noinf}) is satisfied, and the other states (except of the states in the boxes $A_i$ and of the initial state) check if the condition (\ref{inf1}) is satisfied. Reading the input $y$ the automaton first guesses the number $j$ from  condition  (\ref{inf1}) and makes sure that $j$ occurs at least $\obietnic k$ times in $h(y)$. Then it guesses 
$i$ from condition (\ref{inf1}), accepts, and remains in the accepting state forever, unless it spots $\bar{i}$.  This part of the automaton works also correctly for the co-B\"uchi case.
 
The most interesting condition is (\ref{inf0}). It is checked in the following way. At first, the automaton waits in the initial state until it spots $i$ from condition  (\ref{inf0}). Then, it goes to $A_i$, guesses $j$ and goes to the module $B_{i, j}$, which checks if $i$ does not occur any more, and if  both $j$ and $\overline{j}$ occur infinitely often. This can be summarized as: 
 
\begin{theorem} \label{mainth}
Language $L_k$ can be recognized with a nondeterministic B\"uchi automaton with $\Theta(k^2)$ states. 
\end{theorem}
 
 \begin{figure}
  \begin{center}
  \includegraphics[width = 270pt]{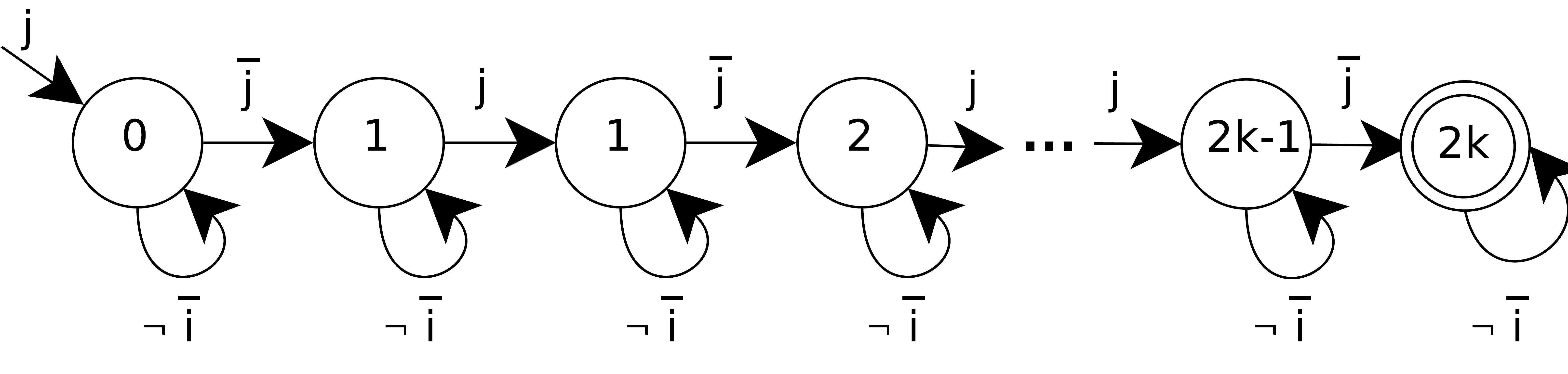}
  \end{center}
  \caption{Automaton $B_{i,j}$ in the NCW case}
 \label{f-bij-cobuchi}
 \end{figure}
 
Condition (\ref{inf0}) cannot be checked by any NCW. However, it can be replaced  by the condition
\begin{eqnarray}
\nonumber i \in \alphabet \wedge w \in (\Sigma_k\setminus\{\overline{i}\})^{\omega} \wedge \exists j \in \alphabet \;\; h_m(viw)=j \text{ for at least } \obietnic k   \text{ numbers } m \label{inf0'} 
\end{eqnarray}
which leads us to a NCW as on Figures \ref{f-NBW}, \ref{f-ai} and \ref{f-bij-cobuchi}. In this case, automaton $B_{i, j}$ needs to count to $\obietnic k$, so it needs $\Theta(k)$ states. Therefore, the whole NCW automaton has $\Theta(k^3)$ states. Actually, we believe that every NCW recognizing $L_k$ indeed needs $\Theta(k^3)$ states. 
 
Now we are ready to state our main theorem:
 
\begin{theorem} \label{main-theorem}
Every NCW recognizing $L_k$ has at least $k \cdot \frac{k^{4/3}}{8}$ states.
\end{theorem}
 
The rest of this paper is devoted to the proof of this theorem. In subsection \ref{disjoint} we will define, for each co-B\"uchi automaton in the normal form, recognizing $L_k$ a family of $k$ disjoint sets of states, and in subsection \ref{mtheorem} we will show that each such set has at least $\frac{k^{4/3}}{4}$ states. As we have seen in subsection \ref{preliminaria}, for a given NCW with $n$ states we can always build a NCW in the normal for with at most $2n$ states, which finally leads to  $\frac 1 2 \cdot k \cdot \frac{k^{4/3}}{4}$ lower bound. 
 
\subsection{The k disjoint sets of states}\label{disjoint}
Let ${\cal A}=\langle \Sigma_k , Q, q_0, \delta , \alpha\rangle$ be an NCW in the normal form with $N$ states that recognizes $L_k$. 
 
Let $w_{i,j} = i(j^N\overline{1},\overline{2},\dots,\overline{i-1}, \overline{i+1}, \dots, \overline{2k})^{\omega}$. 
For every $i\neq j$ let  $q_0$, $q_{i, j}^1$, $q_{i, j}^2$, $q_{i, j}^3, \dots$ be a fixed shortest accepting run of $\cal A$ on $w_{i, j}$.

Words $w_{i,j}$ will be the main tool in our attempt to fool the automaton if it has too few states so let us comment on their structure. First notice, that the $i$, the very first symbol of $w_{i,j}$, will turn into the only $0$ in $h(w)$ -- this is, among other reasons, since for all $m\neq i$ the symbol $\overline{m}$ occurs infinitely many times in $w$. See also that if we replaced
the blocks $j^N$ in the definition of $w_{i,j}$ by just a single $j$, then the word would still be in $L_k$ -- since we do not count promises but fulfillments, the remaining $j$'s are almost redundant. It is only in the proof Lemma \ref{observations}(ii) that we will need them. In the rest of the proof we will only be interested in one state of $A$ per each such block of symbols $j$. 
 For this reason we define  $block(l) =  N + 1 + l(N+2k-1)$ as the function that points to the index of the state in run $q_0, q_{i, j}^1, q_{i, j}^2, q_{i, j}^3, \dots$  just after reading the $l$-th block  $j^N$.

Let $Q_{i, j} = \{q_{i,j}^{block(c)} | c \in {\cal N}\}$.
 
\begin{lemma} \label{lemma1}
For every $i, j, m, l \in \alphabet$ such that $m \neq i \neq j \neq l$ and $m \neq j$, the sets $Q_{i, m}$ and  $Q_{j, l}$ are disjoint. 
\end{lemma}
 
\begin{proof}
 Suppose that there exist $i, j, m, l \in \alphabet$, and $s, t \in {\cal N}$ such that $m \neq i \neq j \neq l$, $m \neq j$ and $q_{i, m}^s = q_{j, l}^t$. Let $v = w_{i, m}[0,  block(s)] . w_{j, l}[block(t) + 1, \infty]$. This word is accepted by $\cal A$, because there exists an accepting run $q_0$, $q_{i,m}^1$, $\dots$ , $q_{i,m}^{block(s)}, q_{j,l}^{block(t) + 1}, q_{j,l}^{block(t) + 2}, \dots$ of $\cal A$. 
 
The only letters without the overline in $v$ are $i$, $m$ and $l$. However, the only overlined letter that does not occur infinitely often in $v$ is $\overline{j}$. This letter is different from $i$, $m$ and $l$ because of the assumptions we made. Therefore  $0$ does not occur in $h(v)$ and $v\not \in L_k$.\eop
\end{proof}
 
We say that $l$ is \emph{huge} if $l>k$ and that $l$ is \emph{small} otherwise.

For every $i$ let $Q_i = \bigcup \{Q_{i, j} | j \text{ is small} \}$. A simple conclusion from Lemma \ref{lemma1} is that for each huge $i, j$ such that $i \neq j$ the sets $Q_i$ and $Q_j$ are disjoint. This implies, that Theorem \ref{main-theorem} will be proved, once we prove the following lemma:
 
\begin{lemma} \label{lemma2}
For each huge $i \in \alphabet$ the size of the set $Q_i$ is greater than $\frac{k^{4/3}}{4}$.
\end{lemma}
 
\subsection{Combinatorial lemma}\label{clemma}
 
The $n \times m$ \emph{state matrix} is a two-dimensional matrix with $n$ rows and $m$ columns. We say that $n \times m$ state matrix is $l$-\emph{painted} if each of its cells is labeled with one of $l$ colors and  the minimal distance between two cells 
in the same row and of the same color is at least $m$.
 
For a  painted $n \times m$ state matrix, we say that an $M_{i, j}$ is a cell \emph{on the left border}  if $j=1$, and is  \emph{on the right border}  if $j = m$. We say that $M_{i, j}$ is \emph{a successor} of $M_{i', j'}$ if $i = i'$ and $j = j'+1$.
 
The \emph{path} $w$ through a painted $n \times m$ state matrix $M$ is a sequence of cells $c_1$, $c_2$, $\dots$, $c_z$ such that $c_1$ is on the left border, $c_z$ is on the right border,  and for each $s<z$ either $c_{s+1}$ is a successor of $c_s$ (we say that ``there is a right move from 
$c_s$ to $c_{s+1}$'') or $c_s$ and $c_{s+1}$ are of the same color (we say that ``there is a jump from 
$c_s$ to $c_{s+1}$'')
 
 
 We say that a path $w$
is \emph{good}, if there are no consecutive $k$ right moves in $w$, and no jump leads to (a cell in) a row that was already visited by this path. Notice that in particular a good path visits at most $k$ cells in any row.

Our main combinatorial tool will be:
 
\begin{lemma} \label{combinatorial-lemma}
 Let $M$ be an $\lfloor \frac {k^{4/3}} {4} \rfloor$-painted $k \times \lfloor \frac{k^{4/3}} {4} \rfloor$ state matrix. Then there exists a good path on $M$. 
\end{lemma} 
 
The proof of this lemma is left to subsection \ref{clemma-proof}
 
\subsection{From automaton to state matrix}\label{mtheorem}
We are  now going to prove Lemma \ref{lemma2}. Let a huge $i \in \alphabet$ be fixed in this subsection
and assume that $|Q_i| < \frac {k^{4/3}} 4$.
  We will show that there exists a word $w$ such that $\cal A$  accepts $w$ and  no agent fulfiles its promises at least $\obietnic k$ times in $w$.
 
Let $j$ be an small number from $\alphabet$. Let us begin from some basic facts about $Q_{i, j}$:

\begin{lemma}\label{observations}
\begin{enumerate}
 \item[(i)] There exists a number $l$ such that for every $s < l$ the state $q_{i, j}^{block(s)}$ is not in $\alpha$ and for every $s\geq l$ the state $q_{i, j}^{block(s)}$ is in $\alpha$. Define $acc(i, j)=l$.
 \item[(ii)] \label{contradiction}  No accepting state from $Q_i$ can be reached on any run of $A$ before some agent  fulfilled its promises $\obietnic k-1$ times.  It also implies that $acc(i, j)\geq \obietnic k -1$.
 \item[(iii)] \label{positional} The states $q_{i, j}^{block(0)}, q_{i, j}^{block(1)}, \dots, q_{i, j}^{block(acc(i, j))}$ are pairwise different. 
\end{enumerate}
\end{lemma}
\begin{proof}
 
\begin{enumerate}

\item[(i)] This is since $\cal A$ is in the normal form.

\item[(ii)] While reading a block of $N$ symbols $j$, the automaton is in $N+1$ states, so there is a state  visited at least twice. If this state was accepting, then a pumping argument would be possible -- we could simply replace the suffix of the word after this block with the word $j^\omega$ and the new word would still  be accepted, despite the fact that it is not in $L_k$.

\item[(iii)] Suppose  $q_{i, j}^{block(s)}$ and  $q_{i, j}^{block(t)}$ are equal and non-accepting. For every $s < t \leq acc(i, j)$, the words $w_{i,j}[block(s)+1, \infty]$ and $w_{i,j}[block(t)+1, \infty]$ are identical. Then a pumping argument works again -- we can find a shorter accepting run by pumping out  the states $q_{i, j}^{block(s)}, \dots, q_{i, j}^{block(t)-1}$. But this contradicts the assumption that our run is shortest.\eop

\end{enumerate}
\end{proof}
We want to show that $|Q_i| \geq \frac {k^{4/3}} 4$. If for any small $j$ there is $acc(i, j) \geq \frac {k^{4/3}} 4 - 1$ then, thanks to Lemma  \ref{observations}(iii) we are done. So, for the rest of this subsection, we assume that
$acc(i, j) < \frac {k^{4/3}} 4 - 1$ for each small $j$.

 
We will now construct a $\lfloor \frac {k^{4/3}} {4} \rfloor$ - painted $k \times \lfloor \frac {k^{4/3}} {4} \rfloor$  state matrix $M$ in such a way, that its $m$'th row will, in a sense, 
represent the accepting run on the word $w_{i,m}$. More precisely, take a $k \times \lfloor \frac {k^{4/3}} {4} \rfloor$   matrix $M$ and
call the cells $M_{m,j}$ of $M$, where $j\leq acc(i,m)$, \emph{ real cells} and call the cells $M_{m,j}$ of $M$ with $j > acc(i,m)$ \emph{ghosts}. For a ghost cell $M_{m,j}$ and the smallest natural number $l$ such that $j-lk\leq acc(i,m)$ call the real
cell $M_(m,j-lk)$ \emph{the host of} $M_{m,j}$. Notice that each ghost has its host, since, by Lemma \ref{observations} (ii),  $acc(i,m) \geq \obietnic k - 1$, which means that there are at least $k$ real cells in each row.

If  $M_{m,j}$ is real then define its color as $q_{i, m}^{block(j-1)}$. If  $M_{m,j}$ is a ghost then define its color as the color of its host.  Now see that $M$ is indeed a $\lfloor \frac {k^{4/3}} {4} \rfloor$ - painted $k \times \lfloor \frac {k^{4/3}} {4} \rfloor$  state matrix --
the condition concerning the shortest distance between cells of the same
color in the same row of $M$ is now satisfied by Lemma \ref{observations} (iii) and  the condition concerning the number of colors is satisfied, since we assume that
$|Q_i| \leq \frac{k^{4/3}}{4}$. 

By Lemma \ref{combinatorial-lemma} we know that there is a good path in $M$. This means that Lemma \ref{lemma2} will 
  be proved once we  show:

\begin{lemma}\label{translating}
 If there exists a good path in $M$, then there exists a word $w \not \in L_k$ such that $w$ is accepted by $\cal A$.
\end{lemma}

\begin{proof}

Suppose $r$ is a good path in $M$ and $c$ is the first ghost cell on $r$.  Let $c'$ be the direct predecessor of $c$ on $r$. If the move from $c'$ to $c$ was a right move then define a new path $p$ as 
the prefix of $r$ ending with $c$. If the move from $c'$ to $c$ was a  jump, then suppose $c''$ is the host of $c$, and define $p$ as the following path: first take the prefix of  $r$ ending with $c'$. Then jump to $c''$ (it is possible, since the color of a ghost is the color of its host). Then make at most $k-1$ right moves to the last real cell in this row.

It is easy to see that $p$ satisfies all the conditions defining a good path, except that it does not reach the right border of $M$.

Let $p$ be a concatenation of words $p_1$,$p_2\ldots$,$p_z$, such that each move between $p_x$ and $p_{x+1}$ is a jump but there are no jumps inside any of $p_x$. This means that
each $p_x$ is contained in some row of $M$, let  $\beta(x)$ be a number of this row. This also means, since $p$ is (almost) a 
good path, that $|p_x|\leq k$  for each $x$.

Let $v_i = \overline{1},\overline{2},\dots,\overline{i-1}, \overline{i+1}, \dots$, $\overline{2k}$. Now define an infinite word $w$ as follows:

$$w = i\beta(1)^N(v_i\beta(1)^N)^{|p_1|-1}  
       (v_i \beta(2)^N)^{|p_2|-1}\ldots 
       (v_i \beta(z)^N)^{|p_z|-1}\beta(z)^\omega  $$

To see that $w\not\in L_k$ notice, that a symbol $s\in \alphabet$ occurs in  $h(w)$ only if $s=\beta(x)$ for some $x\in\{1,2\ldots z\}$ and that it occurs at most $|p_x|+1\leq k$ times in $w$. The fact that $A$ accepts $w$ follows from the construction of path $p$ and from
Lemma \ref{observations} (ii).\eop

\end{proof}

\subsection{Proof of the combinatorial lemma} \label{clemma-proof}

 Let $n=\lfloor \frac {k^{4/3}} 4 \rfloor$ and  $M$ be an $n$-painted $k \times n $ state matrix.
We split the matrix $M$ into matrices $M^0, M^1, ..., M^{\lcell
\frac{2n}{k}\rcell - 1}$,
each of them of $k$ rows and each of them (possibly except of
the last one) of $\frac  k 2$  columns,
such that $M^i$ contains columns $i\frac k 2 +1, i\frac k 2 +2 \dots,
min(i\frac k 2 + \frac k 2, n)$.
The matrices $M^0, M^1, ..., M^{\lcell \frac{2n}{k}\rcell - 2}$ will
be called \emph{multicolumns}.

We are going to build a path $w=c_1c_2\ldots c_z$ through $M$
satisfying the following:
\begin{itemize}
\item if $w$ has a jump from $c_j$ to $c_{j+1}$ then both $c_j$ and
$c_{j+1}$ belong to the same multicomumn;
\item $w$  has exactly ${\lcell \frac{2n}{k}\rcell - 1}$ jumps, one in
each multicolumn;
\item no jump on $w$ leads to a previously visited row of $M$.
\end{itemize}

Clearly, such a path will be a good path. This is since the width  of
each multicolumn is $\frac k 2$, and each sequence of consecutive
right moves  on $w$ will be contained in two adjacent multicolumns
(except of the last such sequence, which is contained in the last
multicolumn
and $M^{\lcell \frac{2n}{k}\rcell - 1}$).

Let $s = \frac {k^{1/3}}{2}$. Since $\lceil s \rceil = \lcell \frac
{2 \cdot k^{4/3}/4}{k} \rcell \geq \lcell \frac{2n}{k}\rcell$,
the number $\lceil s \rceil -1$ is not smaller than the number of jumps we want to make.

Now we concentrate on a single multicolumn $M^i$, which is a matrix
with $k$ rows and with $\frac k 2$ columns. We will call
two rows of such a multicolumn \emph{brothers}  if at least one cell
of one of those rows is of the same color as at least one cell of
another (i.e. two
rows are brothers if a path through $M^i$ can make a jump between them).

Suppose some of the rows of the multicolumn $M^i$ belong to some set
$D^i$ of  \emph{dirty} rows. The rows which are not dirty will be
called \emph{clean}.
A color will be called clean if it occurs in some of the clean rows.
 A row will be called \emph{poor} if it has less than $\lceil s \rceil$ clean
brothers. One needs to take care here -- in the following
procedure, while more rows will get dirty, more rows will also get poor:

\noindent{\bf Procedure} (Contaminate a single multicolumn($D^i$,$M^i$) )\\
\begin{tabular}{c p{10.8cm}}

\hspace*{30pt}	&{\bf while} there are clean poor rows (with respect to the current set $D^i$ of
	dirty rows) in $M^i$, select any clean poor row and all his brothers,
	and make them dirty (changing $D^i$ accordingly).
\end{tabular}\\
{\bf end of procedure}

We would like to know how many new dirty rows can be produced as a
result of  an execution of the above procedure.

Each execution of the body of the while loop makes dirty at most $\lceil s \rceil$
rows and  decreases the number of clean colors
by at least  $\frac k 2$ -- none of the colors of the selected clean
poor row remains clean after the body of the while loop is executed.
Since there are at most $n$ colors in the multicolumn (as $M$ is
$n$-colored), the body of the while loop can be executed at most $\frac n {k/2} \leq \lceil s \rceil$
times, which means that at most $\lceil s \rceil^2$ new dirty rows can be produced.

Notice that after an execution of the procedure, none of the clean rows is poor.

Now we are ready for the next step:

\noindent{\bf Procedure} (Contaminate all multicolumns)\\
\hspace*{30pt}Let $D^{\lcell \frac{2n}{k}\rcell - 1}=\emptyset$;\\
\hspace*{30pt}{\bf for} $i= \lcell \frac{2n}{k}\rcell - 2$ {\bf down to 0}\\
\hspace*{60pt}Let $D^i = D^{i+1}$;\\
\hspace*{60pt}Contaminate a single multicolumn($D^i$,$M^i$);\\
{\bf end of procedure}

We used a convention here, that a set $D^i$ of rows is identified with
the set of numbers of those rows. Thanks to that we could
write the first line of the above procedure, saying ``consider the
dirty rows of $M^{i+1}$ to be also dirty in $M^i$''.

Suppose $D^0,D^1\ldots D^{\lcell \frac{2n}{k}\rcell - 2}$ are sets of
dirty rows in
multicolumns $M^0$,$M^1$, $\ldots$,  $M^{\lcell \frac{2n}{k}\rcell - 2}$
resulting from an execution of the procedure Contaminate all
multicolumns.
Notice, that for each $0\leq i \leq  \lcell \frac{2n}{k}\rcell - 2$
the inclusion  $D^{i+1} \subseteq D^i$ holds. In other words, if a row
is clean in
$M^i$, then it is also clean in  $M^{i+1}$.

The following lemma explains why clean rows are of interest for us:

\begin{lemma}
Suppose $w=c_1c_2\ldots c_z$ is a path through the matrix consisting
of the first $i$ multicolumns of $M$  (or, in other words,
of the first $\frac{ki}{2}$ columns of $M$). Suppose (i) $w$ has
exactly one jump in each multicolumn, and each jump leads to a row which was not visited before, (ii) if there is
a jump from $c_j$ to $c_{j+1}$ then both $c_j$ and $c_{j+1}$ belong to
the same multicomumn. Suppose finally, that
(iii)  the
cell where $w$ reaches the right border of the matrix, belongs to a
clean row $r$. Then $w$ can be extended to a path through the matrix consisting of the first $i+1$ multicolumns of
$M$, in such a way that this extended path will also satisfy
conditions
(i)-(iii).
\end{lemma}

\begin{proof}
 The only thing that needs to be proved is that one can
jump, in multicolumn $M^i$, from row $r$ to some clean row which was
not visited before.
Since, by assumption, $r$ was clean in $M^{i-1}$, it is also clean in
$M^i$. Since there are no  clean poor rows in $M^i$, we know
that $r$ has at least $\lceil s \rceil$ clean brothers. At most $i$ of them were
visited so far by the path, where of course $i \leq \lceil s \rceil-1 $.\eop
\end{proof}

Now, starting from an empty path and a clean row in $M^0$ and using
the above lemma $\lceil \frac {2n} {k}\rceil - 2$ times we can construct a path $w$ as described
in the beginning of this subsection and finish the proof of Lemma \ref{combinatorial-lemma}. The only
lemma we still need for that is:

\begin{lemma}
$|D^0|< k$. In other words, there are clean rows in $M^0$.
\end{lemma}

\begin{proof}
Let $l=\lceil s \rceil - 2$ be the index of the last multicolumn. The number of dirty rows in  $D^{l-i}$ can be bounded by $(i+1)  \cdot \lceil s \rceil^2$ because of observations about defined procedures. For $i=l$, we have $(\lceil s \rceil - 1) \cdot \lceil s \rceil^2$,  what is not greater then $s(s+1)^2=\frac {k^{1/3}} {2}(\frac {k^{1/3}} {2} + 1)^2$ which is, finally, less then $k$, because $k\geq 8$.\eop
\end{proof}

\end{document}